\documentclass[12pt]{article}

\usepackage[margin=1.1in]{geometry}
\usepackage{amsmath,amssymb,amsthm}
\usepackage{booktabs}
\usepackage{array}
\usepackage{enumitem}
\usepackage[authoryear,round]{natbib}
\usepackage[colorlinks=true,linkcolor=blue,citecolor=blue,urlcolor=blue]{hyperref}

\newtheorem{theorem}{Theorem}
\newtheorem{proposition}{Proposition}
\newtheorem{corollary}{Corollary}
\newtheorem{lemma}{Lemma}
\newtheorem{assumption}{Assumption}
\theoremstyle{remark}
\newtheorem{remark}{Remark}

\newcommand{\E}{\mathbb{E}}
\newcommand{\Prob}{\mathbb{P}}
\newcommand{\R}{\mathbb{R}}
\newcommand{\1}{\mathbf{1}}
\newcommand{\calU}{\mathcal{U}}
\newcommand{\calA}{\mathcal{A}}
\newcommand{\calM}{\mathcal{M}}
\newcommand{\calQ}{\mathcal{Q}}
\newcommand{\wbar}{\bar{w}}

\title{Partial Identification from  LLM Prompts}
\author{Xiaohong ~Chen\\ Yale University \and Ashesh Rambachan \\ MIT \and Elie ~Tamer \\ Harvard University}
\date{\today}

\begin{document}
\maketitle

\begin{abstract}
Large language models are increasingly used as binary classifiers when the true label is latent. We study partial identification of the prevalence $(\theta=P(X^*=1))$ from panels of LLM reports whose errors may be arbitrarily dependent given the truth. The design of replication determines the observable, and hence the identifying content: repeated prompts to one model yield a count, several named models a response vector, and both a response matrix. Cast as a two-component finite mixture, the problem makes the identification failure transparent—absent restrictions that separate the latent components, the prevalence
$\theta$ is completely unidentified, and weak stochastic-ordering restrictions (first-order dominance, monotone likelihood ratio, mean ordering) leave the identified set at $[0,1]$. Identifying power comes instead from externally calibrated scores and events, which discipline the mixture in the spirit of the misclassification and corrupted-data literature. We characterize the resulting bounds, establishing validity and sharpness, and give an exact account of the identifying information in the full score distribution beyond its mean. When named models are asked repeated versions of the same question, what identifies $\theta$
 is not the number of positive answers but which models agree across prompts—a feature a vote count discards. An extension  derives implied bounds on regression coefficients when
$X^*$ is a regressor of interest that is not directly observed.
\end{abstract}
\newpage

\baselineskip=1.15\baselineskip
\section{Introduction}

Large language models (LLMs) are now routinely used as binary classifiers. They label text as toxic or non-toxic, factual or non-factual, policy-violating or safe, relevant or irrelevant, and so on. In many applications the target label is not observed in the main sample. The econometrician observes only LLM reports and wants to learn the latent prevalence
\[
\theta=\Prob(X^*=1),
\]
where $X^*\in\{0,1\}$ is the true label.

The central problem is not only that LLMs make mistakes. It is that there is more than one way to replicate an LLM measurement, and different replications produce different observable objects. This paper distinguishes three designs given in Table \ref{tab:designs} below.

\begin{table}[ht]
\centering\small
\caption{Three LLM measurement-panel designs}
\label{tab:designs}
\begin{tabular}{p{3.1cm}p{2.6cm}p{4.2cm}p{4.2cm}}
\toprule
Design & Data for one item & Recommended observable & Reason \\
\midrule
One LLM, repeated questions or prompt variants & $R_1,\dots,R_M$ & Count $S=\sum_{m=1}^M R_m$ & No model identity to preserve; if prompt variants are exchangeable, prompt labels carry no structural content. \\
\addlinespace
Many named LLMs, one question each & $Y=(Y_1,\dots,Y_J)$ & Full named vector $Y\in\{0,1\}^J$ & Model identities matter: LLMs differ in sensitivity, specificity, refusal behavior, and bias. \\
\addlinespace
Many named LLMs, repeated questions each & $R=(R_{jm})_{j\le J,m\le M}$ & Full matrix $R$; under prompt exchangeability, the column-pattern histogram $N$ (lossless); otherwise the model-count vector $T_j=\sum_m R_{jm}$ as a practical coarsening & The matrix is sharpest. $N$ preserves named-model agreement within exchangeable prompts; $T$ preserves model identity while aggregating prompt variation. \\
\bottomrule
\end{tabular}

\smallskip
\footnotesize Notation reserves $J$ for named models and $M$ for repeated prompts. The first row covers the common ``same LLM asked $J$ repeated questions'' case after relabeling that count as $M$.
\end{table}

If the same LLM is asked repeated exchangeable versions of the same binary question, a count is natural. By contrast, if GPT-4, GPT-4 Turbo, Claude, and open-weight models are each asked once, replacing the named vector by a vote count throws away information: the response patterns $(1,1,0,0)$ and $(0,0,1,1)$ have the same count but different implications if the first two models are more reliable. If each named model is asked repeated prompt variants, neither the simple count nor the single-prompt named vector is adequate; the analyst observes a two-way measurement panel and must decide whether model identity, prompt identity, or both should be preserved.

\paragraph{Paper Contributions}
The baseline nonidentification result and the use of sensitivity or specificity-style restrictions to bound a prevalence are applications of partial identification techniques (see \cite{manski1999identification}, \cite{tamer2010partial}). The paper's contribution is to adapt, organize, and extend that literature to handle LLM measurement panels, where errors are plausibly dependent across models and prompts, and to add results that are, to our knowledge, new in this setting:
\begin{enumerate}[leftmargin=1.6em,itemsep=1pt]
\item A design taxonomy (Table~\ref{tab:designs}) that maps the source of replication---repeated prompts, named models, or both---into the correct observable object, with an exact characterization of when a coarsening of the response matrix is \emph{truth-sufficient} (Proposition~\ref{prop:sufficiency}).
\item A unified calibration theory: all of our identifying restrictions are calibrated scores or calibrated events. We prove validity of the resulting bounds, give an exact \emph{sharpness characterization} of the score bounds via a rearrangement function of the observed score distribution (Theorem~\ref{thm:sharp}), show the simple linear bounds are sharp when only the score mean is recorded and exactly sharp for binary scores (Corollary~\ref{cor:sharp}), and prove sharpness of the event bounds.
\item A symmetry result for the two-way design: under prompt exchangeability and an explicit restriction-compatibility condition (Assumption~\ref{ass:compat}), the full matrix has a lossless reduction to the column-pattern histogram $N$ (Theorem~\ref{thm:lossless}). No independence is assumed anywhere.
\item Diagnostics---coarsening loss, row/column influence, dependence audits, and a minimum-tolerance specification test---that add transparency without adding identifying assumptions, together with a multiple-testing-aware treatment of optimized calibrated events.
\end{enumerate}
An augmented empirical illustration quantifies the practical payoff: in a toxicity-labeling application with three named LLMs, reporter-specific calibration on the named vector roughly halves the width of the identified set relative to the count coarsening used in earlier drafts.

\section{Relation to the literature}\label{sec:lit}

We connect our results to important literatures.

\emph{Latent-class and multiple-rater models.} Estimating rater error rates without a gold standard goes back at least to \citet{DawidSkene1979}. That tradition typically obtains point identification through conditional independence of raters given the truth (or low-order dependence corrections). Our setting deliberately drops conditional independence: LLMs share training corpora, benchmarks, synthetic data, distillation pipelines, and alignment procedures, so their errors can be arbitrarily dependent given $X^*$. Without independence, the Dawid--Skene identification route is unavailable, and the model becomes a two-component mixture with unrestricted components---hence partial identification.

\emph{Partial Identification with Misclassification, corrupted data and Mixtures.} Bounds on parameters under misclassified or contaminated outcomes are classical \citep{HorowitzManski1995,Molinari2008,Hu2008}. Our calibrated-score bounds are recognizably of this family: bounds on sensitivity and specificity translate linearly into bounds on prevalence. What we add is (i) the score/event organization that nests reporter-specific accuracy, thresholds, unanimity, weighted ensembles, and matrix events as one assumption type rather than many; and (ii) the sharpness analysis of Section~\ref{sec:calibration}, which distinguishes what is sharp given the score mean from what is sharp given the score law. Finally \citet{HenryKitamuraSalanie2014} study partial identification of finite mixtures using observable variation in mixture weights. Our degeneracy result (Proposition~\ref{prop:degeneracy}) is the mixture-identification observation specialized to LLM panels; its value is the discipline it imposes on applied work, not mathematical novelty. Our group-level extension (Appendix~\ref{app:groups}) connects directly to the exclusion-restriction logic of that literature.

\emph{Classifier calibration and LLM-based labeling.} A growing literature calibrates classifier and LLM confidence and uses calibrated predictions for prevalence estimation \citep{SilvaFilho2023,Hovsepian2024,Multical2026}. We treat external calibration as the \emph{source of identification}: validated lower confidence bounds on score sensitivity/specificity or event predictive values are exactly the inputs our bounds require. The division of labor is deliberate: that literature supplies calibrated constants; this paper says what those constants identify under arbitrary dependence.

\section{Framework and nonidentification}\label{sec:framework}

\subsection{Latent truth, response matrices, and summaries}

For each item, let $X^*\in\{0,1\}$ be the latent truth. Let $j=1,\dots,J$ index named LLMs and $m=1,\dots,M$ index repeated questions, prompt variants, stochastic completions, or elicitation templates. The binary response from model $j$ under prompt $m$ is $R_{jm}\in\{0,1\}$, and the full response matrix is
\[
R=(R_{jm})_{j\le J,\,m\le M}\in\{0,1\}^{J\times M}.
\]
The parameter of interest is $\theta=\Prob(X^*=1)$. In the spirit of the partial identification literature, we develop bounds on $\theta$ using minimal plausible assumptions. 

\  \ \ 

Write $\pi_R(r)=\Prob(R=r)$ and, conditional on the latent state, $f_z(r)=\Prob(R=r\mid X^*=z)$ for $z\in\{0,1\}$. Then
\begin{equation}
\pi_R(r)=(1-\theta)f_0(r)+\theta f_1(r),\qquad r\in\{0,1\}^{J\times M}.
\label{eq:matrixmixture}
\end{equation}
\emph{No factorization of $f_z$ is assumed}: entries of $R$ may be arbitrarily dependent within each latent state.

The analyst may work with a finite summary $U=g(R)\in\calU$: a count, a named response vector, a model-count vector, a prompt-count vector, a column-pattern histogram, or the matrix itself. Let $p_U(u)=\Prob(U=u)$ and $q_{z,U}(u)=\Prob(U=u\mid X^*=z)$. Every summary satisfies
\begin{equation}
p_U(u)=(1-\theta)q_{0,U}(u)+\theta q_{1,U}(u),\qquad u\in\calU.
\label{eq:mixture}
\end{equation}
For restrictions $\calA_U$ on $(q_{0,U},q_{1,U})$, define the identified set
\begin{equation}
\Theta_U(p_U;\calA_U)=\bigl\{\theta\in[0,1]:\exists\,(q_{0,U},q_{1,U})\text{ satisfying \eqref{eq:mixture} and }\calA_U\bigr\}.
\label{eq:idset}
\end{equation}

\ \ 

\subsection{Degeneracy and weak ordering}

\begin{proposition}[Nonidentification of $\theta$]\label{prop:degeneracy}
Fix any finite summary $U=g(R)$. Without restrictions that rule out equality of the latent component distributions, $\Theta_U(p_U)=[0,1]$: for every $\theta\in[0,1]$, the choice $q_{0,U}=q_{1,U}=p_U$ satisfies \eqref{eq:mixture} and the data contain no information about $\theta.$
\end{proposition}

\begin{proof}
For any $u$, $(1-\theta)p_U(u)+\theta p_U(u)=p_U(u)$, for every $\theta\in[0,1]$.
\end{proof}

Weak shape restrictions do not alter this conclusion. For count summaries, first-order stochastic dominance (FOSD), monotone likelihood ratio (MLR) ordering, and weak mean ordering formalize the idea that positive items should produce more positive reports; for vector or matrix summaries, coordinatewise stochastic orders play the same role. These restrictions are often plausible, but they permit $q_0=q_1$, so the degenerate decomposition remains feasible.

For clarity, in a count experiment $S\in\{0,\ldots,M\}$, FOSD means
\[
\sum_{t=s}^M q_1(t)\ge \sum_{t=s}^M q_0(t),\qquad s=1,\ldots,M.
\]
MLR means that $q_1/q_0$ is increasing in the usual cross-product sense: for $s>s'$, $q_1(s)q_0(s')\ge q_1(s')q_0(s)$, with the standard conventions at zeros. Weak mean ordering means $\E[S\mid X^*=1]\ge \E[S\mid X^*=0]$. For vector or matrix summaries, coordinatewise FOSD means $\E[\varphi(U)\mid X^*=1]\ge \E[\varphi(U)\mid X^*=0]$ for every bounded coordinatewise increasing function $\varphi$; equivalently, the inequality holds for all increasing upper sets. All of these are weak orders. Hence $q_0=q_1=p_U$ satisfies them with equality.

\subsection{A confidence-set implication}

The nonidentification result has an inferential counterpart. Let $\calM_U$ be a class of joint laws for $(X^*,U)$. Suppose that for every observable law $p\in\Delta(\calU)$ and every $t$ in a set $\Theta_0\subseteq[0,1]$, the uninformative experiment $U\sim p$, $X^*\sim\mathrm{Bernoulli}(t)$, $X^*\perp U$ belongs to $\calM_U$.

\begin{theorem}[Distribution-free impossibility]\label{thm:impossibility}
Let $\widehat\Theta_{n,\alpha}\subseteq[0,1]$ be a possibly randomized confidence set for $\theta$, constructed from an i.i.d.\ sample $U_1,\dots,U_n$. If
$\sup_{P\in\calM_U}\Prob_P\{\theta(P)\notin\widehat\Theta_{n,\alpha}\}\le\alpha$,
then for every observable law $p$ and every $t\in\Theta_0$, $\Prob_p\{t\notin\widehat\Theta_{n,\alpha}\}\le\alpha$. Consequently
\[
\E_p\bigl[\lambda(\widehat\Theta_{n,\alpha}\cap\Theta_0)\bigr]\ge(1-\alpha)\lambda(\Theta_0),
\]
where $\lambda$ is Lebesgue measure. If $\Theta_0=[0,1]$ then $\E_p[\lambda(\widehat\Theta_{n,\alpha})]\ge 1-\alpha$, and if the confidence set is always an interval contained in $[0,1]$ then $\Prob_p\{\widehat\Theta_{n,\alpha}=[0,1]\}\ge 1-2\alpha$.
\end{theorem}

\begin{proof}
Fix $p$ and $t\in\Theta_0$. The law with $U\sim p$, $X^*\sim\mathrm{Bernoulli}(t)$, $X^*\perp U$ belongs to $\calM_U$; under it the sample has distribution $p^n$ and the prevalence is $t$, so coverage implies $\Prob_p\{t\notin\widehat\Theta_{n,\alpha}\}\le\alpha$. Integrating over $t\in\Theta_0$ (Tonelli) gives the expected-length bound. If $\Theta_0=[0,1]$ and the set is always an interval, coverage of both endpoints implies the interval is $[0,1]$; the final claim follows by the union bound.
\end{proof}

\ \ 

\subsection{Coarsening and information loss}

\begin{theorem}[Coarsening weakens identification under compatible restrictions]\label{thm:coarsening}
Let $U_2=h(U_1)$. Suppose every feasible $(\theta,q_{0,U_1},q_{1,U_1})$ under restrictions $\calA_{U_1}$ projects to a feasible $(\theta,q_{0,U_2},q_{1,U_2})$ under restrictions $\calA_{U_2}$. Then
$\Theta_{U_1}(p_{U_1};\calA_{U_1})\subseteq\Theta_{U_2}(p_{U_2};\calA_{U_2})$.
\end{theorem}

\begin{proof}
Given a feasible point under $U_1$, define $q_{z,U_2}(u_2)=\sum_{u_1:h(u_1)=u_2}q_{z,U_1}(u_1)$. Aggregating \eqref{eq:mixture} over fibers of $h$ gives the mixture equation for $U_2$; compatibility gives the projected restrictions.
\end{proof}

Section~\ref{sec:design3} sharpens this in the matrix design: Proposition~\ref{prop:sufficiency} characterizes exactly when a coarsening is lossless for the truth, and Theorem~\ref{thm:lossless} gives a design symmetry under which a specific coarsening is lossless for the identified set.

\section{Identification by calibration}\label{sec:calibration}

This section contains the paper's maintained identifying content. Everything else---reporter-specific accuracy, threshold classifiers, relaxed support, unanimity, weighted votes, matrix-agreement rules---is a special case obtained by choosing the summary $U$, the score $w$, or the events $A,B$.

\subsection{Calibrated score bounds: validity}

Let $w:\calU\to[0,1]$ be a pre-specified score and write
\[
\wbar=\E[w(U)]=\sum_{u\in\calU}w(u)\,p_U(u).
\]
The score may be a count fraction, a threshold rule, a named-model report, a weighted ensemble, or a matrix-agreement statistic.

\begin{proposition}[Calibrated score bounds]\label{prop:score}
Suppose
\begin{equation}
\E[w(U)\mid X^*=1]\ge a,\qquad \E[1-w(U)\mid X^*=0]\ge b,\qquad a,b\in(0,1].
\label{eq:scorecal}
\end{equation}
Then every admissible prevalence satisfies
\begin{equation}
\max\Bigl\{0,\ \frac{\wbar+b-1}{b}\Bigr\}\ \le\ \theta\ \le\ \min\Bigl\{1,\ \frac{\wbar}{a}\Bigr\}.
\label{eq:scorebounds}
\end{equation}
\end{proposition}

\begin{proof}
Let $\mu_z=\E[w(U)\mid X^*=z]$, so $\wbar=(1-\theta)\mu_0+\theta\mu_1$. Since $\mu_1\ge a$ and $\mu_0\ge 0$, $\wbar\ge\theta a$, giving the upper bound. Since $\mu_0\le 1-b$ and $\mu_1\le 1$, $\wbar\le(1-\theta)(1-b)+\theta=1-b+b\theta$, giving the lower bound. Intersect with $[0,1]$.
\end{proof}

\subsection{Sharpness: score mean vs score law}\label{sec:sharp}

 The interval \eqref{eq:scorebounds} uses only the \emph{mean} $\wbar$ of the score. When the analyst observes the full law of $w(U)$, the sharp identified set can be strictly smaller, and it admits an exact characterization through a rearrangement (concentration) function.

For $t\in[0,1]$ define
\begin{equation}
W^+(t)\ =\ \max\Bigl\{\textstyle\sum_{u}w(u)h(u)\ :\ 0\le h(u)\le p_U(u)\ \forall u,\ \ \textstyle\sum_u h(u)=t\Bigr\}.
\label{eq:Wplus}
\end{equation}
$W^+(t)$ is the largest possible contribution to $\E[w(U)]$ from a sub-population of mass $t$: it is computed by a greedy fill, allocating mass to the values of $u$ with the largest $w(u)$ first (a Hardy--Littlewood rearrangement bound). $W^+$ is concave and nondecreasing, with $W^+(0)=0$ and $W^+(1)=\wbar$, and satisfies $W^+(t)\le\min\{t,\wbar\}$.

\begin{theorem}[Sharp identified set under score calibration]\label{thm:sharp}
Maintain \eqref{eq:scorecal} and suppose the analyst observes $p_U$ (hence the law of $w(U)$) but imposes no other restriction. The identified set for $\theta$ is
\begin{equation}
\Theta_w\ =\ \Bigl\{\theta\in[0,1]\ :\ W^+(\theta)\ \ge\ a\theta\ \ \text{and}\ \ W^+(\theta)\ \ge\ \wbar-(1-b)(1-\theta)\Bigr\},
\label{eq:sharpset}
\end{equation}
and $\Theta_w$ is a (possibly empty) closed interval.
\end{theorem}

\begin{proof}
Work with the joint mass $h_1(u)=\Prob(X^*=1,U=u)$ and $h_0=p_U-h_1$. Feasibility of a prevalence $\theta$ is equivalent to the existence of $h_1$ with
\[
0\le h_1\le p_U,\qquad \sum_u h_1(u)=\theta,\qquad \sum_u w(u)h_1(u)\ge a\theta,\qquad \sum_u w(u)h_0(u)\le(1-b)(1-\theta),
\]
where the third inequality is $\E[w\,\1\{X^*=1\}]\ge a\,\Prob(X^*=1)$, i.e.\ \eqref{eq:scorecal} for $\mu_1$, and the fourth is \eqref{eq:scorecal} for $\mu_0$. The fourth inequality rewrites as $\sum_u w(u)h_1(u)\ge\wbar-(1-b)(1-\theta)$. The feasible set for $h_1$ given the first two constraints is a nonempty polytope (for $\theta\in[0,1]$), and the achievable values of the linear functional $s=\sum_u w(u)h_1(u)$ over that polytope form a closed interval $[W^-(\theta),W^+(\theta)]$. Hence a feasible $h_1$ exists if and only if $W^+(\theta)\ge\max\{a\theta,\ \wbar-(1-b)(1-\theta)\}$. Both $\theta\mapsto a\theta$ and $\theta\mapsto\wbar-(1-b)(1-\theta)$ are affine and $W^+$ is concave, so each constraint defines an interval and $\Theta_w$ is their intersection.
\end{proof}

\begin{corollary}[Relation to the linear bounds; exact sharpness for binary scores]\label{cor:sharp}
\textup{(i)} $\Theta_w$ is contained in the interval \eqref{eq:scorebounds}.
\textup{(ii)} If $w(U)\in\{0,1\}$ almost surely, then $\Theta_w$ equals \eqref{eq:scorebounds}: the linear bounds are exactly sharp for binary scores (in particular, for all event indicators and threshold classifiers).
\textup{(iii)} The interval \eqref{eq:scorebounds} is sharp in the class of observed laws with score mean $\wbar$: for every $\theta$ in \eqref{eq:scorebounds} there exists a law of $w(U)$ with mean $\wbar$ and a feasible decomposition supporting $\theta$. Hence \eqref{eq:scorebounds} cannot be improved using $\wbar$ alone.
\end{corollary}

\begin{proof}
(i) From $W^+(\theta)\le\wbar$ and $W^+(\theta)\ge a\theta$ we get $\theta\le\wbar/a$; from $W^+(\theta)\le\theta$ and $W^+(\theta)\ge\wbar-(1-b)(1-\theta)$ we get $b\theta\ge\wbar+b-1$.
(ii) For binary $w$, $W^+(\theta)=\min\{\theta,\wbar\}$. If $\theta\le\wbar$ the first constraint in \eqref{eq:sharpset} holds since $a\le1$, and the second reduces to $\theta\ge(\wbar+b-1)/b$; if $\theta>\wbar$ the second holds automatically and the first reduces to $\theta\le\wbar/a$. Since $(\wbar+b-1)/b\le\wbar\le\wbar/a$, the union of the two regimes is exactly \eqref{eq:scorebounds}.
(iii) Given $\wbar$, take $w(U)$ binary with $\Prob(w(U)=1)=\wbar$ and apply (ii).
\end{proof}

\begin{remark}[Interpretation]
The gap between $\Theta_w$ and \eqref{eq:scorebounds} is an exact measure of the information in the \emph{shape} of the score distribution beyond its mean. For binary scores the shape carries nothing extra; for graded scores (e.g.\ $w=S/M$ with $M$ large) the rearrangement constraint $W^+(\theta)\ge a\theta$ can bind strictly earlier than $\theta=\wbar/a$, tightening the upper bound. $W^+$ is computed by sorting, so the sharp set costs no more than the linear bounds in practice (Section~\ref{sec:computation}). Emptiness of $\Theta_w$ is a specification rejection of the calibration constants $(a,b)$ at the observed law.
\end{remark}

\subsection{Calibrated events: validity and sharpness}

Let $A\subseteq\calU$ be a high-confidence positive event and $B\subseteq\calU$ a high-confidence negative event.

\begin{proposition}[Posterior event calibration; sharp]\label{prop:event}
If $\Prob(X^*=1\mid U\in A)\ge\rho$ then $\theta\ge\rho\,\Prob(U\in A)$. If $\Prob(X^*=0\mid U\in B)\ge\lambda$ then $\theta\le 1-\lambda\,\Prob(U\in B)$. Each one-sided bound is sharp under the corresponding single event-calibration restriction.
\end{proposition}

\begin{proof}
The lower bound follows from
\[
\theta\ge\Prob(X^*=1,U\in A)=\Prob(X^*=1\mid U\in A)\Prob(U\in A)\ge\rho\Prob(U\in A),
\]
and the upper bound is analogous. For sharpness of the lower bound, under the single restriction involving $A$, set $h_1(u)=\rho p_U(u)$ for $u\in A$ and $h_1(u)=0$ otherwise, with $h_0=p_U-h_1$. Then $0\le h_1\le p_U$, the calibration constraint holds with equality whenever $\Prob(U\in A)>0$, and $\theta=\sum_u h_1(u)=\rho\Prob(U\in A)$. The upper bound is symmetric, assigning mass $h_0(u)=\lambda p_U(u)$ on $B$ and $h_0(u)=0$ off $B$.
\end{proof}

\begin{proposition}[Wrong-state event errors; sharp]\label{prop:eventerror}
If $\Prob(U\in A\mid X^*=0)\le\alpha_A$ with $\alpha_A\in[0,1)$, then
\[
\theta\ \ge\ \max\Bigl\{0,\ \frac{\Prob(U\in A)-\alpha_A}{1-\alpha_A}\Bigr\}.
\]
If $\Prob(U\in B\mid X^*=1)\le\alpha_B$ with $\alpha_B\in[0,1)$, then
\[
\theta\ \le\ \min\Bigl\{1,\ \frac{1-\Prob(U\in B)}{1-\alpha_B}\Bigr\}.
\]
Each one-sided bound is sharp under the corresponding single wrong-state error restriction.
\end{proposition}

\begin{proof}
Let $p_A=\Prob(U\in A)$. Since $p_A\le(1-\theta)\alpha_A+\theta$, rearrangement gives the lower bound. To see sharpness, if $p_A\le\alpha_A$, the bound is zero and is attained by $\theta=0$ with $q_0=p_U$. If $p_A>\alpha_A$, set
\[
\theta_A=\frac{p_A-\alpha_A}{1-\alpha_A}.
\]
Allocate positive-state joint mass only inside $A$, proportionally to $p_U$ on $A$, with total mass $\theta_A$, and set $h_0=p_U-h_1$. Then $h_0(A)=p_A-\theta_A=\alpha_A(1-\theta_A)$, so $\Prob(U\in A\mid X^*=0)=\alpha_A$ and the lower bound is attained. The upper bound is symmetric: if $\Prob(U\in B)\le\alpha_B$, take $\theta=1$; otherwise set $\theta_B=(1-\Prob(U\in B))/(1-\alpha_B)$, allocate all mass outside $B$ to the positive state and allocate additional positive-state mass inside $B$ so that $\Prob(U\in B\mid X^*=1)=\alpha_B$.
\end{proof}

\begin{remark}[One source of identification]\label{rem:onesource}
Propositions~\ref{prop:score}--\ref{prop:eventerror} are the maintained identifying content of the paper. Reporter-specific accuracy, threshold classifiers, relaxed support, unanimity, weighted votes, and matrix-agreement rules are not separate assumptions; they are obtained by choosing different $U$, $w$, $A$, and $B$. Proposition~\ref{prop:eventerror} is the one-sided binary-score analogue of Proposition~\ref{prop:score}: the lower bound uses $w=\1\{U\in A\}$ and the specificity-type constraint $\E[1-w(U)\mid X^*=0]\ge1-\alpha_A$, while the upper bound uses $w=\1\{U\in B\}$ and the false-negative constraint $\E[w(U)\mid X^*=1]\le\alpha_B$. The direct construction above gives sharpness.
\end{remark}

\subsection{Optimized calibrated events and multiplicity}\label{sec:optimized}

Rich summaries admit many candidate high-agreement events. Rather than choosing one arbitrarily, the analyst can pre-specify a finite class and calibrate the whole class on a validation sample, treating event selection explicitly as a multiple-testing problem.

Let $\calA^+$ be a finite class of positive events $A\subseteq\calU$ (row-threshold events, column-pattern events, trusted-model events, weighted-score threshold events). Suppose validation data deliver simultaneous lower confidence bounds $\widehat\rho_A$ such that, with probability at least $1-\alpha$,
\begin{equation}
\Prob(X^*=1\mid U\in A)\ \ge\ \widehat\rho_A\quad\text{for every }A\in\calA^+.
\label{eq:simulcal}
\end{equation}
Then with the same probability all lower bounds hold simultaneously, so
\begin{equation}
\theta\ \ge\ \max_{A\in\calA^+}\ \widehat\rho_A\,\Prob(U\in A),
\label{eq:optlower}
\end{equation}
and symmetrically $\theta\le\min_{B\in\calA^-}\{1-\widehat\lambda_B\Prob(U\in B)\}$ for a simultaneously calibrated negative class $\calA^-$.

\begin{remark}[Pre-specification and honest calibration]
The event class must be fixed before examining the main unlabeled sample, or the calibration step must account for selection. Simultaneity in \eqref{eq:simulcal} can be obtained by Bonferroni corrections, split-sample validation, or conformal-style calibration. The same discipline applies to the calibration constants $(a_j,b_j)$ used anywhere in the paper: for honest inference they should be \emph{lower confidence bounds} estimated on a sample (or sample split) disjoint from the one used to compute observed rates. This adds no structural assumption; it only ensures the validity statements survive the search over events.
\end{remark}

\subsection{Optional sensitivity restrictions}\label{sec:optional-sensitivity}

The paper's baseline identification comes from calibrated scores and calibrated events. Two additional restrictions are useful as sensitivity analyses, but they should not be presented as maintained assumptions unless independently justified. For a count or vote score $C\in\{0,\ldots,J\}$, directional asymmetry with parameter $\gamma>0$ imposes
\[
\E[C\mid X^*=0]\le \gamma\,\E[J-C\mid X^*=1],
\]
so false-positive votes are bounded relative to false-negative votes. For a one-LLM count $S\in\{0,\ldots,M\}$ with observed mean $\bar s=\E[S]$, anchored separation with tolerance $\varepsilon>0$ imposes
\[
\E[S\mid X^*=1]\ge \bar s+\varepsilon,
\qquad
\E[S\mid X^*=0]\le \bar s-\varepsilon,
\]
which implies $\theta\ge \varepsilon/(\varepsilon+M-\bar s)$ and $\theta\le \bar s/(\bar s+\varepsilon)$. These restrictions are informative because they impose cross-state separation, not because they use replication by itself.

\begin{table}[htbp!]
\centering\small
\caption{Restrictions and their role in the paper}
\label{tab:hierarchy}
\begin{tabular}{p{3.0cm}p{5.4cm}p{6.0cm}}
\toprule
Role & Restriction or object & Status \\
\midrule
Regularity & First-order stochastic dominance (FOSD), monotone likelihood ratio (MLR), weak mean ordering & Plausible but non-identifying; optional background restrictions; definitions appear in Section~\ref{sec:framework}. \\
\addlinespace
Main identifying content & Calibrated score $w(U)$ & Core result (Proposition~\ref{prop:score}, Theorem~\ref{thm:sharp}); calibrated from validation data or external evidence. \\
\addlinespace
Main identifying content & Calibrated event $A\subseteq\calU$ or $B\subseteq\calU$ & Core result (Propositions~\ref{prop:event}--\ref{prop:eventerror}); includes high- and low-agreement regions; sharp. \\
\addlinespace
Design structure & Exchangeability of prompt labels, when built into the design & Not an accuracy assumption; justifies lossless reduction from $R$ to the orbit statistic $N$ (Theorem~\ref{thm:lossless}). \\
\addlinespace
Special cases & Reporter-specific accuracy, threshold rules, relaxed support, unanimity, weighted votes & Instances of score or event calibration, not separate assumptions. \\
\addlinespace
Optional sensitivity & Directional asymmetry, anchored separation, conditional independence & Not baseline; define explicitly, use only when substantively justified, and report separately. \\
\addlinespace
Diagnostics & Coarsening loss, row/column influence, dependence audits, minimum tolerance & Empirical transparency tools; add no identifying power. \\
\bottomrule
\end{tabular}
\end{table}

\section{Design I: one LLM, repeated binary questions}\label{sec:design1}

A single LLM is used repeatedly to measure the same latent binary truth. For one item, write the repeated reports as $R_1,\dots,R_M\in\{0,1\}$. If the repeated prompts are exchangeable probes of the same truth, their labels carry no structural content and the natural observable is the count
\[
S=\sum_{m=1}^M R_m\in\{0,\dots,M\},\qquad
p(s)=(1-\theta)q_0(s)+\theta q_1(s).
\]
No conditional independence is imposed across prompts: repeated prompts from one model may share the same systematic errors.

\begin{remark}[When the count is appropriate]
The count is appropriate only if the repetitions are exchangeable measurements of the same latent truth. If prompts ask substantively different questions, there is no single scalar $X^*$ behind all responses. If prompt identities have known reliability differences, keep the prompt-response vector rather than counting.
\end{remark}

\noindent The count design has two natural calibrated scores: $w_1(S)=S/M$ and $w_k(S)=\1\{S\ge k\}$.

\paragraph{Average repeated-prompt accuracy.} With $w=S/M$ in Proposition~\ref{prop:score} and $\bar r=\E[S/M]$,
\begin{equation}
\max\Bigl\{0,\frac{\bar r+b-1}{b}\Bigr\}\le\theta\le\min\Bigl\{1,\frac{\bar r}{a}\Bigr\}.
\label{eq:design1mean}
\end{equation}
Because $S/M$ is a graded score, Theorem~\ref{thm:sharp} applies with content: the sharp set computed from $W^+$ can be strictly inside \eqref{eq:design1mean}, and is obtained by sorting the support of $S$.

\paragraph{Threshold accuracy.} For a threshold $k$, let $D_k=\1\{S\ge k\}$ and $p_k^+=\Prob(S\ge k)$. If validation data support $\Prob(D_k=1\mid X^*=1)\ge a_k$ and $\Prob(D_k=0\mid X^*=0)\ge b_k$, then
\[
\max\Bigl\{0,\frac{p_k^++b_k-1}{b_k}\Bigr\}\le\theta\le\min\Bigl\{1,\frac{p_k^+}{a_k}\Bigr\},
\]
and by Corollary~\ref{cor:sharp}(ii) these bounds are exactly sharp.

\paragraph{High- and low-count events.} With $A=\{S\ge k\}$ and $B=\{S\le\ell\}$, Proposition~\ref{prop:event} gives $\theta\ge\rho_k\Prob(S\ge k)$ and $\theta\le 1-\lambda_\ell\Prob(S\le\ell)$; Proposition~\ref{prop:eventerror} gives the tail-error bounds
\[
\theta\ \ge\ \max\Bigl\{0,\frac{\Prob(S\ge k)-\alpha_{0k}}{1-\alpha_{0k}}\Bigr\},\qquad
\theta\ \le\ \min\Bigl\{1,\frac{1-\Prob(S\le\ell)}{1-\alpha_{1\ell}}\Bigr\}.
\]
Relaxed support is the special case $k=M$, $\ell=0$: $q_0(M)\le\alpha_0$, $q_1(0)\le\alpha_1$. Exact support ($\alpha_0=\alpha_1=0$) gives $p(M)\le\theta\le 1-p(0)$ but is rarely credible for LLMs; relaxed, validation-calibrated tail bounds are usually preferable.

\paragraph{Implication.} Design I is useful only to the extent that the count distribution can be calibrated. Report a small number of validation-calibrated score or event bounds---ideally the sharp set of Theorem~\ref{thm:sharp} for the graded score---rather than many sensitivity restrictions.

\section{Design II: many named LLMs, one binary question each}\label{sec:design2}

$J$ named LLMs each answer the same binary question once. The response vector is $Y=(Y_1,\dots,Y_J)\in\{0,1\}^J$ with law $\pi(y)$ and class-conditionals $f_z(y)$ satisfying
$\pi(y)=(1-\theta)f_0(y)+\theta f_1(y)$.
No conditional independence is assumed across named LLMs. The vote count $\sum_j Y_j$ is a coarsening, not the primitive object.

The named-vector design matters because it permits calibrated scores and events that use model identity: the single-reporter score $w_j(Y)=Y_j$, the weighted score $w(Y)=\sum_j c_jY_j$ with $c_j\ge0$, $\sum_j c_j=1$, and events $A\subseteq\{0,1\}^J$ encoding agreement by a trusted subset rather than simple majority.

\paragraph{Reporter-specific accuracy.} Suppose validation data provide
\begin{equation}
\Prob(Y_j=1\mid X^*=1)\ge a_j,\qquad \Prob(Y_j=0\mid X^*=0)\ge b_j,\qquad j=1,\dots,J,
\label{eq:reportercal}
\end{equation}
and let $m_j=\Prob(Y_j=1)$. Applying Proposition~\ref{prop:score} to each $w_j(Y)=Y_j$ and intersecting,
\begin{equation}
\max_{1\le j\le J}\max\Bigl\{0,\frac{m_j+b_j-1}{b_j}\Bigr\}\ \le\ \theta\ \le\ \min_{1\le j\le J}\min\Bigl\{1,\frac{m_j}{a_j}\Bigr\}.
\label{eq:reporterbounds}
\end{equation}
Each one-reporter bound is sharp (binary score); the intersection is valid and is the sharp set based on the marginals alone. Using the joint law $\pi$ with the restrictions \eqref{eq:reportercal} simultaneously can tighten further; this is the LP of Section~\ref{sec:computation}.

\paragraph{Named events.} If $A\subseteq\{0,1\}^J$ satisfies $\Prob(X^*=1\mid Y\in A)\ge\rho_A$ then $\theta\ge\rho_A\pi(A)$; if $B$ satisfies $\Prob(X^*=0\mid Y\in B)\ge\lambda_B$ then $\theta\le 1-\lambda_B\pi(B)$. Relaxed unanimity is only the special case $A=\{\mathbf 1_J\}$, $B=\{\mathbf 0_J\}$ and should not be the default event unless validation evidence supports it.

\paragraph{The exact cost of counting votes.} If only the vote count $S=\sum_jY_j$ is stored, the reporter-specific marginals and named events are not recoverable; the count identifies only $\bar m=\frac1J\E[S]=\frac1J\sum_j m_j$. With averaged calibration constants $\beta_1=\frac1J\sum_j a_j$, $\beta_0=\frac1J\sum_j b_j$, the count-only analogue is
\begin{equation}
\max\Bigl\{0,\frac{\bar m+\beta_0-1}{\beta_0}\Bigr\}\ \le\ \theta\ \le\ \min\Bigl\{1,\frac{\bar m}{\beta_1}\Bigr\},
\label{eq:countbounds}
\end{equation}
which is generally strictly weaker than \eqref{eq:reporterbounds}; Section~\ref{sec:empirical} quantifies the gap in the toxicity application. The same loss applies to weighted scores and trusted-subset events.

\paragraph{Implication.} Analyze Design II with the named vector. Count-based analysis is a robustness check or a fallback when only counts were stored. The identifying content should come from reporter-specific or named-event calibration, not from generic independence assumptions.

\section{Design III: many named LLMs, repeated questions each}\label{sec:design3}

The third design is a two-way measurement panel: for one item, observe $R=(R_{jm})\in\{0,1\}^{J\times M}$ with mixture law \eqref{eq:matrixmixture}. No conditional independence is assumed across rows, columns, or cells. Design III is valuable not because it makes independence credible, but because it records \emph{where} agreement occurs.

\subsection{Summaries and coarsenings}

Useful lower-dimensional summaries: the model-count vector $T=(T_1,\dots,T_J)$, $T_j=\sum_m R_{jm}$; the prompt-count vector $S=(S_1,\dots,S_M)$, $S_m=\sum_j R_{jm}$; the total count $C=\sum_{j,m}R_{jm}$; and, when prompt labels are exchangeable but model labels are not, the \emph{column-pattern histogram}
\begin{equation}
N_y(R)=\sum_{m=1}^M\1\{R_{\cdot m}=y\},\qquad y\in\{0,1\}^J,\qquad N(R)=(N_y(R):y\in\{0,1\}^J),
\label{eq:histogram}
\end{equation}
which counts the prompts on which the named response pattern equals $y$ (with $J=3$, $N_{110}$ counts prompts where models 1 and 2 say one and model 3 says zero). The histogram preserves cross-model agreement within prompts while discarding prompt labels, and refines the model-count vector via $T_j=\sum_y y_jN_y$. The coarsening hierarchy is
\[
R\ \longrightarrow\ N\ \longrightarrow\ T\ \longrightarrow\ C,\qquad R\ \longrightarrow\ S\ \longrightarrow\ C,
\]
with support sizes $2^{JM}$, $\binom{M+2^J-1}{2^J-1}$, $(M+1)^J$, $(J+1)^M$, and $JM+1$ respectively. The full matrix distinguishes patterns with the same total count: one weak model saying yes on every prompt is not equivalent to several trusted models each saying yes repeatedly.

\subsection{Truth-sufficient reductions of the response matrix}

Which coarsenings lose information about $X^*$? The answer is a likelihood-ratio invariance condition.

Let $\Omega=\{0,1\}^{J\times M}$, let $U=h(R)$ be any finite summary with fibers $F_u=\{r:h(r)=u\}$, and let $q_z(u)=\sum_{r\in F_u}f_z(r)$.

\begin{proposition}[Truth-sufficient coarsenings]\label{prop:sufficiency}
Assume $0<\Prob(X^*=1)<1$. The following are equivalent, up to null sets.
\begin{enumerate}[leftmargin=1.6em,itemsep=1pt]
\item $U$ is sufficient for the latent truth: $\Prob(X^*=1\mid R=r)=\Prob(X^*=1\mid U=h(r))$ for all $r$ with positive probability.
\item The residual matrix information given the summary is independent of the truth: $\Prob(R=r\mid U=u,X^*=1)=\Prob(R=r\mid U=u,X^*=0)$ for $r\in F_u$.
\item The likelihood ratio $L(r)=f_1(r)/f_0(r)$ is constant on each fiber $F_u$ (usual conventions for zeros).
\item There is a Markov kernel $K(r\mid u)$, independent of $z$, with $f_z(r)=q_z(h(r))\,K(r\mid h(r))$ for $z=0,1$.
\end{enumerate}
When these conditions hold, $R$ and $U$ contain the same information about $X^*$; when they fail, the coarsening discards information about the latent truth.
\end{proposition}

\begin{proof}
(2)$\Leftrightarrow$(4): take $K(r\mid u)=\Prob(R=r\mid U=u,X^*=z)$, which is independent of $z$ exactly when (2) holds. (2) is equivalent to $f_1(r)/q_1(u)=f_0(r)/q_0(u)$ on $F_u$, i.e.\ to $f_1(r)/f_0(r)=q_1(u)/q_0(u)$ constant on $F_u$, which is (3). By Bayes' rule, $\Prob(X^*=1\mid R=r)=\theta f_1(r)/\{(1-\theta)f_0(r)+\theta f_1(r)\}$ depends on $r$ only through $h(r)$ iff $f_1/f_0$ does, giving (1)$\Leftrightarrow$(3).
\end{proof}

\begin{remark}[Common matrix summaries]
$C$ is truth-sufficient only if $f_1/f_0$ is constant across matrices with the same total count; $T$ only if constant across matrices with the same row-count vector; $N$ only if constant across matrices with the same pattern histogram. Prompt exchangeability of both $f_0$ and $f_1$ is sufficient for the last property but not necessary: likelihood-ratio invariance within prompt-permutation orbits is enough.
\end{remark}

\subsection{Prompt exchangeability and a lossless matrix reduction}\label{sec:lossless}

Some reductions are not losses: they remove labels that have no design meaning. In Design III, prompt labels may be exchangeable even when model labels are not.

Let $S_M$ be the group of permutations of prompt labels; for $\sigma\in S_M$, $(\sigma r)_{jm}=r_{j,\sigma(m)}$. The histogram $N(R)$ indexes the orbits of this action: two matrices have the same $N$ iff one is obtained from the other by permuting prompt columns. The action extends to distributions by $(\sigma q)(r)=q(\sigma^{-1}r)$.

The reduction from $R$ to $N$ is lossless only when prompt exchangeability is part of the maintained design. To avoid ambiguity, let $\Theta_R^{\mathrm{ex}}(\pi_R;\calA_R)$ denote the full-matrix identified set when admissible component distributions $q_0,q_1$ are required to be prompt-exchangeable in addition to satisfying the matrix-level restrictions $\calA_R$. The required correspondence between matrix-level and histogram-level restrictions is explicit below.

\begin{assumption}[Restriction compatibility]\label{ass:compat}
\textup{(i)} The matrix-level restrictions $\calA_R$ are invariant to prompt relabeling: if $(q_0,q_1)\in\calA_R$ and the $q_z$ are prompt-exchangeable, then relabeling prompts does not change the truth of the restrictions.
\textup{(ii)} The histogram-level restriction set is exactly the projection of the exchangeable part of $\calA_R$:
\[
\calA_N=\Bigl\{(q_0^N,q_1^N):\ (q_0,q_1)\in\calA_R,\ q_z(\sigma r)=q_z(r)\ \forall\sigma\in S_M,\ q_z^N\ \text{is the law of }N\text{ under }q_z\Bigr\}.
\]
\end{assumption}

\begin{theorem}[Lossless reduction under prompt exchangeability]\label{thm:lossless}
Suppose prompt exchangeability is a maintained design restriction, the observed population law $\pi_R$ is prompt-exchangeable, and Assumption~\ref{ass:compat} holds. Then the full matrix and the column-pattern histogram give the same identified set:
\[
\Theta_R^{\mathrm{ex}}(\pi_R;\calA_R)\ =\ \Theta_N(p_N;\calA_N).
\]
\end{theorem}

\begin{proof}
($\subseteq$) Let $(\theta,q_0,q_1)$ be feasible for $\Theta_R^{\mathrm{ex}}(\pi_R;\calA_R)$. Projecting each $q_z$ along $N$ gives $(q_0^N,q_1^N)$ satisfying the histogram mixture equation because projection commutes with mixing. Assumption~\ref{ass:compat}(ii) gives $(q_0^N,q_1^N)\in\calA_N$.

($\supseteq$) Let $(\theta,q_0^N,q_1^N)$ be feasible for $\Theta_N(p_N;\calA_N)$. For each histogram value $n$, lift $q_z^N(n)$ uniformly over the finite orbit $\{r:N(r)=n\}$. The lifted $q_z$ are prompt-exchangeable and project back to $q_z^N$. Because $\pi_R$ is prompt-exchangeable, it is itself the uniform orbit lift of $p_N$, so the lifted distributions satisfy the matrix mixture equation for $\pi_R$. Assumption~\ref{ass:compat}(ii) ensures that the lifted pair satisfies $\calA_R$. Thus the same $\theta$ is feasible for $\Theta_R^{\mathrm{ex}}(\pi_R;\calA_R)$.
\end{proof}

\begin{remark}[No independence is used]
The theorem uses only a design symmetry and matched restrictions. It does not assume prompts are independent conditional on $X^*$, nor that cells are weakly correlated; entries of $R$ may be arbitrarily dependent within each latent state.
\end{remark}

\begin{remark}[Why compatibility matters]\label{rem:compatfails}
Both directions of Assumption~\ref{ass:compat} have bite. If the analyst imposes \emph{additional} restrictions after reducing to $N$ (so $\calA_N$ is strictly smaller than the projection), then only $\Theta_N\subseteq\Theta_R^{\mathrm{ex}}$ is guaranteed. Conversely, a matrix-level restriction that is \emph{not} prompt-invariant---for example, a calibrated accuracy bound for prompt $m=1$ only---cannot be expressed through $N$ at all; dropping it in the reduction gives only $\Theta_R^{\mathrm{ex}}\subseteq\Theta_N$. Equality is a statement about matched restriction classes, not about the statistic alone.
\end{remark}

\begin{remark}[General orbit reduction]
The argument applies to any finite group of design symmetries: if $G$ acts on the matrix space, admissible $q_z$ are $G$-invariant, and the restrictions are $G$-invariant and matched as in Assumption~\ref{ass:compat}, the orbit statistic gives the same identified set as the full matrix. Prompt exchangeability gives $N$; independent prompt relabeling within each model gives $T$; full exchangeability of all cells gives $C$, but that last symmetry is rarely credible for named LLMs.
\end{remark}

\subsection{Matrix scores and matrix events}

Design III should not be analyzed as a large vote count unless both dimensions are intentionally treated as exchangeable. The natural calibrated objects are matrix scores and matrix events.

A weighted matrix score has the form
\[
w(R)=\frac{\sum_{j}\sum_{m}c_{jm}R_{jm}}{\sum_{j}\sum_{m}c_{jm}},\qquad c_{jm}\ge0,\ \ \sum_{j,m}c_{jm}>0,
\]
with special cases $w=T_j/M$, $w=S_m/J$, and $w=C/(JM)$; if \eqref{eq:scorecal} holds for $w$, Proposition~\ref{prop:score} and Theorem~\ref{thm:sharp} apply. Under prompt exchangeability a score may be written as a function of $N$, e.g.\ placing weight on column patterns where trusted models agree.

Matrix events encode repeated agreement across both dimensions, e.g.
\[
A^+(K,L)=\Bigl\{R:\ \sum_{j=1}^J\1\{T_j\ge L\}\ \ge K\Bigr\},\qquad
A^-(K,L)=\Bigl\{R:\ \sum_{j=1}^J\1\{T_j\le L\}\ \ge K\Bigr\},
\]
the events that at least $K$ named models each produce at least (at most) $L$ positive prompt responses. The histogram also suggests events invisible from $T$: with a trusted subset $J_0\subseteq\{1,\dots,J\}$,
\[
\sum_{y:\,y_j=1\ \forall j\in J_0} N_y\ \ge\ L
\]
requires the trusted models to agree positively on at least $L$ prompts, regardless of the other models, while remaining invariant to prompt relabeling. Calibrated predictive values or wrong-state errors for any of these events feed Propositions~\ref{prop:event}--\ref{prop:eventerror}; searches over event classes use Section~\ref{sec:optimized}.

\subsection{Model-specific and prompt-specific bounds; structured calibration}

The model-count vector is a tractable compromise when $R$ or $N$ is too large. If for each named model $j$
\[
\E[T_j/M\mid X^*=1]\ge a_j,\qquad \E[(M-T_j)/M\mid X^*=0]\ge b_j,
\]
then with $r_j=\E[T_j/M]$,
\[
\max_{1\le j\le J}\max\Bigl\{0,\frac{r_j+b_j-1}{b_j}\Bigr\}\ \le\ \theta\ \le\ \min_{1\le j\le J}\min\Bigl\{1,\frac{r_j}{a_j}\Bigr\},
\]
and symmetrically for calibrated prompt-count scores $S_m/J$. Model-specific bounds suit stable, calibratable model error profiles; prompt-specific bounds suit stable prompt families. If prompt labels are exchangeable but within-prompt agreement matters, $N$ is preferable to $T$ because it retains more agreement geometry.

Full cell-level calibration of every pair $(j,m)$ may be too parameter-rich. A practical compromise groups models and prompts: with model groups $g(j)$ and prompt families $h(m)$, calibrate
$\Prob(R_{jm}=1\mid X^*=1)\ge a_{g(j),h(m)}$ and $\Prob(R_{jm}=0\mid X^*=0)\ge b_{g(j),h(m)}$,
with the grouping fixed before analysis or justified by validation data.

\subsection{Diagnostics unique to the matrix design}

Design III allows diagnostics that add no identifying assumptions.

\paragraph{Coarsening loss.} For $U\in\{R,N,T,S,C\}$ and compatible restrictions, report widths of identified sets and their differences, e.g.
\[
\mathrm{Loss}(N\to T)=\mathrm{wid}(\Theta_T)-\mathrm{wid}(\Theta_N),\qquad
\mathrm{Loss}(T\to C)=\mathrm{wid}(\Theta_C)-\mathrm{wid}(\Theta_T).
\]
A large loss from $N$ to $T$ means within-prompt cross-model agreement matters; from $T$ to $C$, that model identity matters. If prompt exchangeability is credible and restrictions are invariant and matched, Theorem~\ref{thm:lossless} predicts no loss from $R$ to $N$ at the population level---a checkable implication of the design symmetry.

\paragraph{Row and column influence.} Let $\Theta^{(-j)}$ and $\Theta^{(-m)}$ be the identified sets after dropping model $j$ or prompt $m$. Large movements in either bound show that conclusions hinge on a particular model or prompt; this is often more informative than another structural assumption.

\paragraph{Dependence and redundancy.} On a validation sample, estimate within-model dependence $\mathrm{Corr}(R_{jm},R_{jm'}\mid X^*=z)$ and across-model dependence $\mathrm{Corr}(R_{jm},R_{j'm}\mid X^*=z)$. High correlations indicate redundant measurements; low correlations indicate nonredundant variation. These are diagnostics, not independence assumptions.

\section{Computation and inference}\label{sec:computation}

\paragraph{Linear programming.} For any finite summary $U$, fixed-$\theta$ feasibility is a linear program. With joint masses $h_z(u)=\Prob(X^*=z,U=u)$:
\[
h_0(u)+h_1(u)=p_U(u),\qquad \sum_u h_1(u)=\theta,\qquad h_z\ge0,
\]
and conditional restrictions become linear after multiplying through, e.g.\ $\E[w(U)\mid X^*=1]\ge a$ becomes $\sum_u w(u)h_1(u)\ge a\sum_u h_1(u)$. Sharp bounds under the baseline linear score and event restrictions are two LPs: minimize and maximize $\sum_u h_1(u)$. For fixed $\theta$, FOSD restrictions are linear in the conditional component probabilities and can be included in fixed-$\theta$ feasibility checks; exact MLR restrictions are bilinear in the components and should be handled only through explicit relaxations or nonlinear optimization. Prompt exchangeability is implemented either by using $N$ directly or by orbit-equality constraints on matrix probabilities.

\paragraph{The sharp score set by sorting.} $W^+(t)$ in \eqref{eq:Wplus} is computed greedily: order support points by descending $w(u)$ and fill mass to total $t$. $W^+$ is piecewise linear and concave, so $\Theta_w$ in \eqref{eq:sharpset} is found by intersecting a concave piecewise-linear function with two affine functions---no LP solver needed.

\paragraph{Sampling uncertainty.} When $p_U$ is estimated from $N$ items, replace mixture equalities by bands $|(1-\theta)q_0(u)+\theta q_1(u)-\widehat p_U(u)|\le\Delta_u$. A simple finite-sample choice is the Hoeffding-union bound
\begin{equation}
\varepsilon_N(\alpha)=\sqrt{\frac{\log(2|\calU|/\alpha)}{2N}},
\label{eq:hoeffding}
\end{equation}
with $\Delta_u=\varepsilon_N(\alpha)$ for every $u$, yielding a conservative outer confidence set; less conservative alternatives use the multinomial bootstrap, empirical Bernstein bands, or moment-inequality methods. For optimized event bounds, sampling uncertainty enters twice---through $\widehat p_U$ in the main sample and through validation estimates of predictive values---and validity requires the simultaneous calibration \eqref{eq:simulcal}.

\paragraph{Specification testing by minimum tolerance.} Let $\calQ$ be the set of distributions over $\calU$ implied by the maintained restrictions for some $\theta$, and define
\[
\Delta^*=\min_{q\in\calQ}\|\widehat p_U-q\|_\infty,
\]
the minimum $\ell_\infty$ tolerance restoring feasibility. If the restrictions are correct at the population law $p_U^0$, then $\Delta^*\le\|\widehat p_U-p_U^0\|_\infty$, so a finite-sample test rejects at level $\alpha$ if $\Delta^*>\varepsilon_N(\alpha)$, and
\[
\Delta_0\in\bigl[\max\{0,\Delta^*-\varepsilon_N(\alpha)\},\ \Delta^*+\varepsilon_N(\alpha)\bigr]
\]
is a confidence interval for the population misspecification distance. The test detects out-of-distribution collapse: if an LLM panel becomes uninformative, restrictions calibrated in distribution may become infeasible. Emptiness of the sharp score set $\Theta_w$ is the same idea specialized to a single calibrated score.

\section{Simulation: count-based Beta-Binomial experiment}\label{sec:simulation}

Set $M=5$ and $\theta_0=0.65$, with class-conditional counts $S\mid X^*=1\sim\mathrm{BetaBin}(5,4,1.5)$ and $S\mid X^*=0\sim\mathrm{BetaBin}(5,1.5,4)$; the observed histogram is $p=(1-\theta_0)q_0+\theta_0q_1$. The design has $\E[S\mid X^*=0]=1.36$, $\E[S\mid X^*=1]=3.64$, average repeated-prompt accuracy $0.727$, and observed mean $\E[S]=2.84$, so $\wbar=\E[S/M]=0.568$ for the graded score $w=S/M$. Without restrictions, or under FOSD/MLR/weak mean ordering alone, the identified set is $[0,1]$ because the degenerate decomposition $q_0=q_1=p$ remains feasible (Proposition~\ref{prop:degeneracy}); Table~\ref{tab:simsets} records this once as a benchmark and then isolates the paper's main computational comparison: the mean-only linear bounds \eqref{eq:scorebounds} versus the sharp rearrangement set of Theorem~\ref{thm:sharp}, computed from $W^+$ by sorting the six support points of $S$.

\begin{table}[ht]
\centering\small
\caption{Simulation DGP: class-conditional and mixture PMFs}
\label{tab:simdgp}
\begin{tabular}{cccc}
\toprule
$s$ & $q_0(s)$ & $q_1(s)$ & $p(s)$ \\
\midrule
0 & 0.310 & 0.015 & 0.118 \\
1 & 0.291 & 0.055 & 0.137 \\
2 & 0.208 & 0.121 & 0.152 \\
3 & 0.121 & 0.208 & 0.178 \\
4 & 0.055 & 0.291 & 0.208 \\
5 & 0.015 & 0.310 & 0.207 \\
\bottomrule
\end{tabular}
\end{table}

\begin{table}[ht]
\centering\small
\caption{Simulation identified sets for $\theta$: linear (mean-only) versus sharp (full-law) score bounds}
\label{tab:simsets}
\setlength{\tabcolsep}{4pt}
\begin{tabular}{llcccc}
\toprule
Restriction & Uses & $\theta_L$ & $\theta_U$ & Width & Covers $\theta_0$ \\
\midrule
None, or FOSD/MLR/mean order & -- & $\approx0$ & $\approx1$ & 1.000 & Yes \\
Exact support, $p(5)\le\theta\le1-p(0)$ & tail cells & 0.207 & 0.882 & 0.675 & Yes \\
Linear score \eqref{eq:scorebounds}, $a=b=0.60$ & mean $\wbar$ & 0.280 & 0.947 & 0.667 & Yes \\
Sharp score (Thm.~\ref{thm:sharp}), $a=b=0.60$ & law of $w(S)$ & 0.317 & 0.947 & 0.630 & Yes \\
Linear score \eqref{eq:scorebounds}, $a=b=0.70$ & mean $\wbar$ & 0.383 & 0.812 & 0.429 & Yes \\
Sharp score (Thm.~\ref{thm:sharp}), $a=b=0.70$ & law of $w(S)$ & 0.479 & 0.784 & 0.305 & Yes \\
Sharp score + exact support, $a=b=0.60$ & law + support & 0.317 & 0.882 & 0.565 & Yes \\
Sharp score + exact support, $a=b=0.70$ & law + support & 0.479 & 0.784 & 0.305 & Yes \\
\bottomrule
\end{tabular}

\smallskip
\footnotesize Notes: the score is $w=S/M$ with $\wbar=0.568$. Sharp sets are computed from the rearrangement function $W^+$ of \eqref{eq:Wplus} by sorting; the combined sharp-score-plus-support rows are computed by the LP of Section~\ref{sec:computation} and agree with the closed form of Theorem~\ref{thm:sharp} when support does not bind.
\end{table}

The table mirrors the paper's identification message in miniature. Weak ordering restrictions leave $[0,1]$ untouched, so they are recorded in a single benchmark row. Calibration is what moves the set, and the shape of the score distribution carries identifying content beyond its mean, exactly as Theorem~\ref{thm:sharp} predicts: at $a=b=0.60$ the rearrangement constraint $W^+(\theta)\ge \wbar-(1-b)(1-\theta)$ binds before the linear lower bound, raising $\theta_L$ from $0.280$ to $0.317$; at $a=b=0.70$ the sharp set $[0.479,0.784]$ removes $29\%$ of the linear interval's width ($0.429$ to $0.305$), tightening both endpoints. Because $w(S)$ is a graded score on six support points, this gain is obtained by a sort and costs nothing computationally. Exact support tightens only the weakly calibrated case ($a=b=0.60$, where it caps $\theta_U$ at $0.882$) and is redundant once calibration is strong; this ordering---calibration first, support as a benchmark---is the recommended reporting style of Section~\ref{sec:recommendations}.

\section{Empirical illustration: toxicity classification}\label{sec:empirical}

This section revisits the toxicity-response classification task studied by \cite{ChengEtAl2024SoftLabel}. The task is to bound the prevalence of toxic responses when the analyst observes binary labels from LLM annotators and, for a validation sample, expert-adjudicated truth (which gives us a benchmark to allow us to validate our bounds).

In particular, the data involve $J=3$ named LLM annotators---GPT-4, GPT-4 Turbo, and Claude-2---and three human annotators. This is a Design II setting.  Following the taxonomy of this paper, we now report the named-vector analysis in three increasingly disciplined steps---plug-in marginal bounds, the full named-vector LP on the joint law of $Y=(Y_1,Y_2,Y_3)$, and an honest version that replaces plug-in calibration constants by one-sided lower confidence bounds---and we tag every bound with the response object that must be stored to compute it.

\begin{table}[ht]
\centering\small
\caption{Annotator characteristics on validation set, $N=1{,}000$, $\theta_0=0.650$}
\label{tab:annotators}
\begin{tabular}{lcccc}
\toprule
Annotator & Positive rate & Sensitivity & Specificity & Accuracy \\
\midrule
Human 1 & 0.536 & 0.786 & 0.929 & 0.836 \\
Human 2 & 0.627 & 0.857 & 0.800 & 0.837 \\
Human 3 & 0.601 & 0.852 & 0.866 & 0.857 \\
GPT-4 & 0.562 & 0.795 & 0.871 & 0.822 \\
GPT-4 Turbo & 0.483 & 0.662 & 0.849 & 0.727 \\
Claude-2 & 0.272 & 0.345 & 0.863 & 0.526 \\
LLM exchangeable average & -- & 0.601 & 0.861 & -- \\
\bottomrule
\end{tabular}
\end{table}

Table~\ref{tab:annotators} shows why named-vector analysis matters: GPT-4 is much more accurate than Claude-2, and all three LLMs have higher specificity than sensitivity. A count-only analysis collapses this heterogeneity into an exchangeable average. The validation class-conditionals also show that exact support is empirically violated with only three LLMs: among truly toxic items, 14.2\% receive zero positive LLM votes; among truly non-toxic items, 4.3\% receive unanimous positive votes. Exact support can be reported as a benchmark, but relaxed support or validation-calibrated score bounds are more credible.

Table~\ref{tab:histograms} records the observable objects that the bounds below are built from. The top panel gives the count histograms---$\widehat p(s)$ with the validation class-conditionals $\widehat q_z(s)$ for $N=1{,}000$, and the count-only training histogram for $N=28{,}194$---while the lower panel reports the full named-vector law $\widehat\pi(y)$ over $\{0,1\}^3$. The two are not interchangeable: the count is the coarsening $s=\sum_j y_j$ of $\widehat\pi$, so patterns that disagree on \emph{which} model flagged the item are merged. Here $(1,1,0)$, where GPT-4 and GPT-4 Turbo flag and Claude-2 abstains, has mass $0.250$, whereas the equal-count pattern $(0,1,1)$ has mass only $0.010$; counting pools them into $\widehat p(2)=0.295$ and discards exactly the reporter identity that the named-vector bounds will exploit.

\begin{table}[ht]
\centering\small
\caption{Observed response objects, $J=3$ LLMs: vote-count histograms and named-vector law}
\label{tab:histograms}
\begin{tabular}{ccccc@{\qquad}cc}
\toprule
\multicolumn{5}{c}{Validation set, $N=1{,}000$} & \multicolumn{2}{c}{Training set, $N=28{,}194$} \\
\cmidrule(r){1-5}\cmidrule(l){6-7}
$s$ & $\widehat p(s)$ & $\widehat q_0(s)$ & $\widehat q_1(s)$ & & $s$ & Train $\widehat p(s)$ \\
\midrule
0 & 0.348 & 0.731 & 0.142 & & 0 & 0.218 \\
1 & 0.172 & 0.163 & 0.177 & & 1 & 0.162 \\
2 & 0.295 & 0.063 & 0.420 & & 2 & 0.177 \\
3 & 0.185 & 0.043 & 0.262 & & 3 & 0.444 \\
$\bar p$ & 1.317 & -- & -- & & $\bar p$ & 1.847 \\
$\E[S\mid X^*]$ & -- & 0.417 & 1.802 & & & \\
\bottomrule
\end{tabular}

\medskip
\setlength{\tabcolsep}{4pt}
\begin{tabular}{lcccccccc}
\multicolumn{9}{l}{\emph{Named-vector law $\widehat\pi(y)$, $y=(y_{\text{GPT-4}},y_{\text{Turbo}},y_{\text{Cl-2}})$, validation:}}\\
\toprule
$y$ & 000 & 001 & 010 & 011 & 100 & 101 & 110 & 111 \\
$\widehat\pi(y)$ & 0.348 & 0.042 & 0.038 & 0.010 & 0.092 & 0.035 & 0.250 & 0.185 \\
\bottomrule
\end{tabular}

\smallskip
\footnotesize Notes: the count histogram $\widehat p(s)$ is the coarsening of $\widehat\pi(y)$ along $s=\sum_j y_j$; the patterns $(1,1,0)$ and, e.g., $(0,1,1)$ are merged by counting even though they involve reporters of very different accuracy. For the training set only the count histogram was stored.
\end{table}

\subsection{Named-vector bounds: plug-in benchmark and honest calibration}\label{sec:namedempirical}

Table~\ref{tab:namedbounds} reports reporter-specific bounds \eqref{eq:reporterbounds} in two versions. Panel~A is the \emph{plug-in benchmark}: calibration constants $(a_j,b_j)$ are the validation sensitivities and specificities of Table~\ref{tab:annotators}, estimated on the same sample as the observed rates $m_j$, so the panel illustrates the identification logic but is not honest inference. Panel~B is the \emph{honest calibration} version recommended in Section~\ref{sec:optimized}: $(a_j,b_j)$ are replaced by one-sided Clopper--Pearson lower confidence bounds $(a_j^L,b_j^L)$ at Bonferroni level $\alpha/6$ with $\alpha=0.05$, so all six constraints hold simultaneously with probability at least $0.95$ and the resulting bounds are valid by Proposition~\ref{prop:score}.

\begin{table}[ht]
\centering\small
\caption{Reporter-specific bounds, validation set: plug-in benchmark versus honest calibration}
\label{tab:namedbounds}
\begin{tabular}{lccccc}
\toprule
Score & $m_j$ & $a_j$ & $b_j$ & $\theta_L$ & $\theta_U$ \\
\midrule
\multicolumn{6}{l}{\emph{Panel A: plug-in benchmark (constants and rates from the same sample)}} \\
GPT-4 ($w=Y_1$) & 0.562 & 0.795 & 0.871 & 0.497 & 0.707 \\
GPT-4 Turbo ($w=Y_2$) & 0.483 & 0.662 & 0.849 & 0.391 & 0.730 \\
Claude-2 ($w=Y_3$) & 0.272 & 0.345 & 0.863 & 0.156 & 0.789 \\
Intersection & & & & \textbf{0.497} & \textbf{0.707} \\
\midrule
\multicolumn{6}{l}{\emph{Panel B: honest calibration (one-sided 95\% simultaneous lower confidence bounds)}} \\
GPT-4 ($w=Y_1$) & 0.562 & 0.755 & 0.823 & 0.468 & 0.744 \\
GPT-4 Turbo ($w=Y_2$) & 0.483 & 0.615 & 0.797 & 0.351 & 0.785 \\
Claude-2 ($w=Y_3$) & 0.272 & 0.300 & 0.813 & 0.105 & 0.906 \\
Intersection & & & & \textbf{0.468} & \textbf{0.744} \\
\bottomrule
\end{tabular}

\smallskip
\footnotesize Notes: bounds use \eqref{eq:reporterbounds}. In Panel B, $a_j^L$ and $b_j^L$ are exact one-sided Clopper--Pearson lower confidence bounds for sensitivity ($n_1=650$) and specificity ($n_0=350$) at level $0.05/6$ each. Both panels cover $\theta_0=0.650$; honest calibration widens the intersection from $[0.497,0.707]$ to $[0.468,0.744]$ (width $0.209$ to $0.277$)---the price of validity is modest. Fully honest inference would additionally compute $(a_j^L,b_j^L)$ on a split disjoint from the sample used for $m_j$ and $\widehat\pi$, and add sampling bands \eqref{eq:hoeffding}.
\end{table}

\subsection{The key comparison: what storing richer objects buys}

Table~\ref{tab:keycomparison} is the section's central exhibit. It reports four identified sets for the same population, ordered by the richness of the stored response object, including the full named-vector LP of Section~\ref{sec:computation}, which imposes the mixture equation on the joint law $\widehat\pi(y)$ over $\{0,1\}^3$ together with all reporter-specific calibration restrictions simultaneously.

\begin{table}[ht]
\centering\small
\caption{Key comparison: identified sets for $\theta$ by stored response object, validation set}
\label{tab:keycomparison}
\setlength{\tabcolsep}{5pt}
\begin{tabular}{llccc}
\toprule
Bound & Object stored & $\theta_L$ & $\theta_U$ & Width \\
\midrule
Count coarsening, eq.~\eqref{eq:countbounds} & count $S$ & 0.348 & 0.731 & 0.383 \\
Reporter-specific marginal intersection & named vector $Y$ (marginals) & 0.497 & 0.707 & 0.209 \\
Full named-vector LP & joint law of $Y$ & 0.497 & 0.707 & 0.209 \\
Honest full named-vector LP & joint law of $Y$ & 0.468 & 0.744 & 0.277 \\
\bottomrule
\end{tabular}

\smallskip
\footnotesize Notes: all rows cover $\theta_0=0.650$. Row 1 uses the exchangeable-average constants $(\beta_1,\beta_0)=(0.601,0.861)$ with $\bar m=0.439$. Rows 2--3 use the plug-in constants of Table~\ref{tab:namedbounds}, Panel A; row 4 uses the honest constants of Panel B. The honest count-coarsened analogue of row 1 (averaging the $a_j^L,b_j^L$) is $[0.308,0.788]$, width $0.480$: the storage gain survives honest calibration, $0.277$ versus $0.480$.
\end{table}

Three facts stand out. First, the cost of counting: moving from the named vector to the count widens the set from $[0.497,0.707]$ to $[0.348,0.731]$---the width nearly doubles, from $0.209$ to $0.383$---purely because counting discards reporter identity, quantifying Theorem~\ref{thm:coarsening}. Second, the full LP on the joint law coincides with the marginal intersection to four decimals. This is informative rather than disappointing: it certifies that, given reporter-specific calibration alone, the marginal intersection is already sharp at this observed law, so no further information can be extracted from these restrictions; the joint law is nevertheless the object to store, because only it makes the sharpness check possible and because pattern-level calibrated events (trusted-subset agreement, Section~\ref{sec:optimized}) require it whenever validation evidence supports them. Third, honest calibration costs about $0.07$ of width relative to the plug-in benchmark but preserves the storage ranking: the honest named-vector set (width $0.277$) remains far narrower than the honest count set (width $0.480$).

The binding reporter on both sides is GPT-4, the most accurate model; the influence diagnostic of Section~\ref{sec:design3} applies verbatim: dropping GPT-4 widens the plug-in intersection to the Turbo bounds $[0.391,0.730]$. For the training set ($N=28{,}194$), only the count histogram was stored, so every named-vector row of Table~\ref{tab:keycomparison} is unavailable by construction: with $\bar m=1.847/3=0.616$ and the exchangeable-average constants, \eqref{eq:countbounds} gives $[0.554,\,1.000]$. The storage choice has real consequences: on the training set the count bound is $[0.554,1.000][0.554, 1.000]
[0.554,1.000]$, whose upper endpoint is the trivial value 
1, so counting leaves the prevalence bounded only from below. Had the named vector been stored, both endpoints would be informative.

\section{Extension: regression on a latent label}\label{sec:regression}

So far the target has been the scalar prevalence $\theta=\Prob(X^*=1)$. In many applications $X^*$ is not the final object but a latent \emph{regressor}: the analyst wants the partial association between an observed outcome and the latent label, holding covariates fixed. This section shows that the calibration bounds feed directly into such a regression, with prevalence recovered as the special case of a constant outcome.

For each item let $(V,Z,U)$ be observed, where $V\in\R$ is a scalar outcome, $Z\in\R^{d}$ a vector of observed covariates, and $U=g(R)$ the LLM summary of Sections~\ref{sec:framework}--\ref{sec:design3}; $X^*\in\{0,1\}$ remains latent. (We write the outcome as $V$ to avoid collision with the named-LLM vector $Y$ of Section~\ref{sec:design2}.) Stack the regressors as $W=(1,Z',X^*)'\in\R^{d+2}$ and consider the best linear predictor (BLP) coefficient
\[
\beta \;=\; \big(\E[WW']\big)^{-1}\E[WV],
\]
assuming $\E[WW']$ is nonsingular. The coefficient on $X^*$, denoted $\gamma$, is typically the parameter of interest; if $\E[V\mid Z,X^*]$ is linear, $\beta$ is the conditional-expectation coefficient.

\paragraph{Reduction to latent cross-moments.} Because $(V,Z)$ are observed and $X^*$ is binary (so $(X^*)^2=X^*$), every entry of $\E[WW']$ and $\E[WV]$ is identified from the data except the blocks that pair $X^*$ with an observed variable:
\[
\theta=\E[X^*],\qquad \E[X^*Z]\in\R^{d},\qquad \E[X^*V]\in\R.
\]
Collect these in $\mu=(\theta,\E[X^*Z]',\E[X^*V])'$. Then $\beta=\Phi(\mu;m_{\mathrm{obs}})$ for a known map $\Phi$ that is rational in $\mu$---a matrix inverse times a vector, both affine in $\mu$---where $m_{\mathrm{obs}}$ collects the observed moments. Identifying $\beta$ thus reduces to identifying $\mu$.

\paragraph{Bounding the latent cross-moments.} Since $(U,Z,V)$ are jointly observed, write the unknown attachment of the latent label to the data as
\[
\eta(u,z,v)=\Prob(X^*=1\mid U=u,Z=z,V=v)\in[0,1].
\]
Every latent cross-moment is a \emph{linear} functional of $\eta$: for any observed $h$,
\[
\E[X^*\,h(U,Z,V)]=\E[h(U,Z,V)\,\eta(U,Z,V)].
\]

\begin{lemma}[Calibrated bounds on latent moments]\label{lem:latentmoments}
Let $h$ be any observed, bounded function. Over all $\eta$ consistent with the observed law of $(U,Z,V)$ and with the maintained calibration restrictions of Section~\ref{sec:calibration}, the cross-moment $\E[X^*h]$ ranges over a closed interval whose endpoints solve the linear programs
\[
\min_{\eta}\ /\ \max_{\eta}\ \E[h(U,Z,V)\,\eta(U,Z,V)]
\quad\text{s.t.}\quad 0\le\eta\le1,\ \text{marginal consistency},\ \text{calibration}.
\]
Taking $h\equiv1$ returns the prevalence bounds of Section~\ref{sec:calibration} exactly; taking $h\in\{Z_1,\dots,Z_d,V\}$ bounds the remaining components of $\mu$.
\end{lemma}

The calibration inequalities enter exactly as in Section~\ref{sec:computation}: a score restriction $\E[w(U)\mid X^*=1]\ge a$ becomes $\sum_u w(u)\,\E[\eta\mid U=u]\,p_U(u)\ge a\,\E[\eta]$ after writing $\E[X^*w(U)]=\E[w(U)\eta]$, and an event restriction becomes the corresponding linear inequality; both are linear in $\eta$. The lemma is therefore the paper's fixed-$\theta$ LP with the objective $\E[\eta]$ replaced by $\E[h\,\eta]$.

\begin{proposition}[Identified set for the regression coefficient]\label{prop:regbounds}
Let $\mathcal M\subseteq\R^{d+2}$ be the set of $\mu$ attainable by some feasible $\eta$; $\mathcal M$ is convex and compact, and is a polytope under score/event (linear) calibration. The sharp identified set for the BLP coefficient is the image
\[
B=\{\Phi(\mu;m_{\mathrm{obs}}):\mu\in\mathcal M\}.
\]
Each coordinate of $B$ is computed by the fixed-value method of Section~\ref{sec:computation}. By Frisch--Waugh--Lovell, $\gamma=\E[\tilde V\,\eta]\,/\,D(\theta,\E[X^*Z])$, where $\tilde V$ is the residual of $V$ on $(1,Z)$ (observed) and $D=\E[\tilde X^{*2}]\ge0$ depends only on the first-stage moments $(\theta,\E[X^*Z])$. Fixing those first-stage moments at any value in their calibrated region makes $\gamma=c$ linear in $\eta$, so $B$ is traced by a parametric family of LPs indexed by that low-dimensional region.
\end{proposition}

Replacing $\mathcal M$ by the box of component-wise intervals from Lemma~\ref{lem:latentmoments} gives valid but generally conservative outer bounds; the joint program is sharp. This is the same distinction drawn for the named-vector LP versus the marginal intersection in Section~\ref{sec:empirical}.

\begin{remark}[Conditional versus marginal calibration]
The bounds are valid under the calibration of Section~\ref{sec:calibration}, which constrains $\eta$ only through its $U$-marginal and so permits $\eta$ to vary freely with $(Z,V)$ within report cells. If instead calibration is validated \emph{within} covariate strata---$\E[w(U)\mid X^*=1,Z=z]\ge a(z)$, the natural design when validation data carry $Z$---then within-stratum prevalences and outcome cross-moments are bounded separately and the coefficient bounds tighten accordingly. Conditional calibration is to this section what reporter-specific calibration was to Design~II: it is where the covariates earn their keep.
\end{remark}

\begin{remark}[Relation to misclassified-regressor bounds]
With $U$ a single noisy report, this is the partial-identification problem for a regression with a misclassified binary regressor \citep{Bollinger1996,Mahajan2006,Hu2008,Molinari2008}. Our contribution is not a new bound for that problem but the observation that the \emph{same} externally calibrated scores and events that bound prevalence also bound the regression coefficient, through the single channel of the latent cross-moments $\mu$; no independence or instrument is invoked.
\end{remark}

\paragraph{A transparent special case.} With no covariates and target the mean contrast $\gamma=\E[V\mid X^*=1]-\E[V\mid X^*=0]$,
\[
\gamma=\frac{\E[VX^*]-\theta\,\E[V]}{\theta(1-\theta)},
\]
so $\gamma$ is pinned down once $\theta$ and $\E[VX^*]$ are, each bounded by Lemma~\ref{lem:latentmoments}; the identified set is obtained by ranging the numerator and $\theta$ jointly over their calibrated region. If, in addition, $V$ is itself a calibrated high-confidence label, the contrast inherits the sharpness of Section~\ref{sec:calibration}.

\paragraph{Inference.} Sampling uncertainty enters through the observed law of $(U,Z,V)$ and, for honest calibration, through the validation estimates of the calibration constants. Propagating the bands \eqref{eq:hoeffding} of Section~\ref{sec:computation} (or a multiplier bootstrap) through the LPs of Lemma~\ref{lem:latentmoments} yields a valid outer confidence set for each coordinate of $\beta$; moment-inequality methods apply verbatim because all restrictions are linear in $\eta$.

\section{Practical recommendations}\label{sec:recommendations}

\paragraph{Use calibration as the main identifying content.} The maintained assumptions should be calibrated score and event restrictions. Report the source of calibration, preferably lower confidence bounds from validation data on a split disjoint from the main sample.

\paragraph{Treat weak ordering as regularity, not identification.} FOSD, MLR, and weak mean ordering are plausible but do not rule out $q_0=q_1$ and should not be presented as identifying assumptions.

\paragraph{Report sharp sets where they are cheap.} For binary scores and events, the linear bounds are already sharp. For graded scores, report the sharp set of Theorem~\ref{thm:sharp}; it costs a sort, and the gap from the linear bounds measures the information in the score's shape.

\paragraph{Exploit design symmetries, not independence assumptions.} If prompt labels are exchangeable by design, reduce the matrix to the column-pattern histogram $N$ under matched, invariant restrictions (Assumption~\ref{ass:compat}) rather than imposing conditional independence.

\paragraph{Store the richest object possible.} For Design I, store all repeated responses even if the count is analyzed. For Design II, store the named vector. For Design III, store the full matrix whenever feasible; under prompt exchangeability, store or construct $N$; at minimum store the model-count vector. The training-set panel of Section~\ref{sec:empirical} shows what is lost otherwise.

\paragraph{Report coarsening and influence diagnostics.} Compare bounds under richer and coarser summaries; report whether results depend on particular models or prompts; quantify losses along $R\to N\to T\to C$ where relevant.

\paragraph{Keep sensitivity assumptions separate.} Directional asymmetry, anchored separation, and independence-style restrictions can be useful, but label them as sensitivity analysis unless independently justified.

\section{Conclusion}\label{sec:conclusion}

The identifying content of LLM measurement panels depends on the source of replication. Asking one LLM repeated exchangeable questions creates a count experiment; asking several named LLMs once creates a named-reporter experiment; asking several named LLMs repeated questions creates a two-way response matrix. These designs should not be collapsed into a single vote-count model at the outset.

The general theory is simple, and we present it as an adaptation of known mixture and misclassification logic to a setting where dependence among reporters is the rule. Every finite summary yields a two-component mixture for the latent truth; without restrictions, prevalence is completely unidentified, and weak ordering restrictions do not help because they permit equality of the latent components. Useful identification comes from calibrated scores and calibrated events, of which reporter-specific accuracy, threshold rules, relaxed support, unanimity, weighted ensembles, and matrix-agreement events are special cases. The sharpness analysis clarifies exactly what each calibrated object delivers: linear bounds that are sharp for binary scores and for any analysis recording only the score mean, and a rearrangement characterization of the strictly sharper set available from the full score law.

The matrix design adds one further lesson. Even when entries of the matrix are arbitrarily correlated, design symmetries justify lossless reductions: under prompt exchangeability and matched invariant restrictions, the column-pattern histogram preserves all identifying information in the full matrix. Calibration identifies; storage and symmetry determine what can be calibrated without loss.

\appendix

\section{Multiple observable groups}\label{app:groups}

Suppose items belong to observable groups $W\in\{1,\dots,G\}$ with known shares $\rho_g$ and group-specific observed histograms $p_g(s)=\Prob(S=s\mid W=g)$, each group has prevalence $\pi_g=\Prob(X^*=1\mid W=g)$, and the latent class-conditional count distributions are common across groups:
\[
p_g(s)=(1-\pi_g)q_0(s)+\pi_g q_1(s),\qquad g=1,\dots,G,
\]
with target $\theta=\sum_g\rho_g\pi_g$. The common-$q_z$ restriction is an exclusion restriction---groups shift prevalence but not measurement technology---in the spirit of \citet{HenryKitamuraSalanie2014}, and should be used only when measurement errors are plausibly stable across groups. The exact model is bilinear in $(\pi_g,q_z)$; a conservative LP is obtained from McCormick envelopes for the products $\pi_gq_1(s)$ and $(1-\pi_g)q_0(s)$.

\newpage 
\bibliographystyle{apalike}

\bibliography{references}
\end{document}